\def\draft{0}
\def\llncs{0}
\def\anon{0}

\ifnum\llncs=0
\documentclass[11pt]{article}
\usepackage[in]{fullpage}
\usepackage[utf8]{inputenc}
\else
\documentclass[orivec,runningheads,envcountsect,envcountsame]{llncs}
\fi

\newcommand{\boldpar}[1]{\vspace{3pt}\par\noindent\textbf{#1}}

\ifnum\llncs=1
\renewcommand{\paragraph}{\boldpar}
\fi

\usepackage[shortlabels]{enumitem}
\usepackage{bm}
\usepackage{amsfonts}
\usepackage{amsmath}
\usepackage{amssymb}
\usepackage{xcolor}
\usepackage{mathtools}
\usepackage{graphicx}
\usepackage{float}
\usepackage{hyperref}
\usepackage[nameinlink]{cleveref}
\usepackage{braket}
\usepackage{mathrsfs}
\usepackage{tikz}
\usepackage{qcircuit}
\usepackage{dsfont}
\usepackage{authblk}

\usepackage[style=alphabetic,minalphanames=3,maxalphanames=4,maxnames=99,backref=true]{biblatex}

\DeclareFieldFormat{eprint:iacr}{Cryptology ePrint Archive: \href{https://ia.cr/#1}{\texttt{#1}}}
\DeclareFieldFormat{eprint:iacrarchive}{Cryptology ePrint Archive: \href{https://eprint.iacr.org/archive/#1}{\texttt{#1}}}
\hypersetup{colorlinks={true},linkcolor={blue},citecolor=magenta}

\ifnum\llncs=0

\usepackage{amsthm}
\usepackage{amssymb}

\theoremstyle{plain}
\newtheorem{theorem}{Theorem}[section]
\newtheorem{lemma}[theorem]{Lemma}

\newtheorem{corollary}[theorem]{Corollary}

\newtheorem{proposition}[theorem]{Proposition}
\newtheorem{question}{Question}
\newtheorem{definition}[theorem]{Definition}
\newtheorem{remark}[theorem]{Remark}

\fi

\newcommand{\ketbra}[2]{\ket{#1}\!\bra{#2}}

\renewcommand{\cal}[1]{\mathcal{#1}}

\newcommand{\C}{\mathbb{C}}
\newcommand{\iu}{\mathrm{i}\mkern1mu}
\newcommand{\N}{\mathbb{N}}

\newcommand{\E}{\mathop{\mathbb{E}}}
\newcommand{\Tr}{\mathrm{Tr}}

\newcommand{\id}{\mathds{1}}
\newcommand{\eps}{\epsilon}
\newcommand{\binset}{\{0,1\}}

\newcommand{\cA}{\cal{A}}
\newcommand{\cB}{\cal{B}}
\newcommand{\cD}{\cal{D}}

\newcommand{\rnd}[1]{\lceil {#1} \rfloor}
\newcommand{\abs}[1]{\left| {#1} \right|}
\newcommand{\norm}[1]{\left\| {#1} \right\|}

\newcommand{\dnorm}[1]{\left\| {#1} \right\|_{\diamond}}

\newcommand{\poly}{\mathrm{poly}}
\newcommand{\polylog}{\mathrm{polylog}\, }

\newcommand{\loglog}{\mathrm{loglog}\, }

\title{On Scalable Pseudorandom Unitaries and the Unitary Synthesis Problem}

\ifnum\llncs=0

\author[1]{Zvika Brakerski}
\author[2]{Henry Yuen}

\affil[1]{Weizmann Institute of Science}
\affil[2]{Columbia University}

\else

\ifnum\anon=0
\author{Zvika Brakerski\inst{1} \and Henry Yuen\inst{2}}
\institute{Weizmann Institute of Science \and Columbia University}
\else
\author{}
\institute{}
\fi

\fi

\newcommand{\net}{\mathscr{N}}
\newcommand{\unitary}{\mathrm{U}}
\newcommand{\cvg}{\mathscr{C}}
\newcommand{\xps}{\eta}

\newcommand{\vol}{\mathsf{vol}}
\newcommand{\udesign}{\nu}
\newcommand{\hr}{\mathsf{Haar}}
\newcommand{\support}{\mathsf{support}}
\newcommand{\secp}{\kappa}
\newcommand{\qb}{\tau}

\begin{document}

\date{}

\maketitle

\begin{abstract}
We consider the task of constructing pseudorandom unitaries (PRUs) with \emph{scalable security}, i.e.\ families in which the security parameter may vary independently of the dimension (or input bit-length). It is not known whether scalable PRUs can be constructed. In this work we show that if scalable PRUs can be constructed via the prevailing paradigm for analyzing PRUs, then there would be a positive solution to the Aaronson-Kuperberg unitary synthesis problem, a longstanding question in quantum complexity theory about whether implementing arbitrary unitaries can be efficiently reduced to computing a Boolean function.

Specifically, we formalize the notion of \emph{ROM-PRUs}, which are statistically secure PRUs in the random oracle model (ROM). All prior known constructions of cryptographically secure PRUs are based on a ROM-PRU construction. We prove novel connections between ROM-PRUs, approximate unitary designs, $\eps$-nets over the unitary group, and the unitary synthesis problem. In particular, we prove that any unitary synthesis algorithm (and thus any ROM-PRU) must use a classical oracle with input length $(2 - o(1)) \log d$ bits, where $d$ is the dimension of the unitary to be implemented. This bound rules out \emph{all existing candidates} for scalable PRUs in the literature.

Together, these connections indicate that ROM-PRUs provide a fruitful idealized model for studying pseudorandom unitaries.

\end{abstract}

\ifnum\draft=1
\setcounter{tocdepth}{4}
\tableofcontents
\fi

\section{Introduction}
\label{sec:intro}

A pseudorandom unitary (PRU) ensemble~\cite{JLS18} is a family of efficiently implementable unitaries $\{U_s\}_s$ such that, for a random $s$, oracle access to $U_s$ is indistinguishable from oracle access to a Haar-random unitary. Whereas \cite{JLS18} defined this object and proposed a few candidates, only recently was it shown by \cite{MPSY24,MH25} that PRUs can be constructed. These constructions, as well as \emph{essentially all} known constructions of quantum pseudorandom objects from cryptographic assumptions, follow the ``ROM paradigm''. Namely, the PRU is first constructed and shown to be information-theoretically secure in the random oracle model (ROM), and then the oracle is replaced by a cryptographic pseudorandom function (PRF) to obtain a construction in the standard model under the assumption that post-quantum one-way functions exist.
In particular, \cite{MPSY24,MH25} constructed such \emph{ROM-PRUs} and proved that an adversary making $t$ queries to the PRU has statistical advantage $\lesssim t^2/d$ in distinguishing the construction from a Haar-random unitary, where $d$ is the dimension of the Hilbert space of the unitary (i.e.\ the unitary acts on $n = \log d$ qubits). This means that the proof provides no guarantee whenever the number of queries roughly exceeds $\sqrt{d}$. When the number of qubits $n$ is large, e.g.\ larger than the security parameter, the value $t^2/d$ is negligible.

What happens when $n$ is small, or otherwise one wishes to reduce the error further? This property is known as \emph{scalability} \cite{BS20}. It was established in \cite{MPSY24} that sequential repetition decreases the distinguishing gap, so it is possible to reduce the statistical advantage to $\lesssim (t^2/d)^k$ for any polynomial $k$. However, even in this case, once $t \gtrsim \sqrt{d}$, there is no security guarantee whatsoever.

We therefore define the \emph{query bound} $\qb$ of a ROM-PRU construction as the value of $t$ for which the statistical advantage equals $1/2$. We note that the choice of constant here is arbitrary, in particular in light of the aforementioned sequential amplification technique. This is because a construction with advantage $< 1/2^k$ can be derived from one with advantage $< 1/2$, for any polynomial $k$. A scalable ROM-PRU should therefore be able to support arbitrarily high $\qb = \poly(d, 2^\secp)$, where $\secp$ is the security parameter.

We posit the importance of studying the query bound of ROM-PRU constructions, and specifically the maximal achievable $\qb$. In particular, we show that it has implications for the possibility of a \emph{unitary synthesis} algorithm. Unitary synthesis addresses the question of whether the hardness of implementing arbitrary quantum unitaries can be reduced to the hardness of implementing arbitrary classical functions. Specifically, it asks whether there exists an efficient quantum reduction $R$ such that, for any $n$-qubit unitary $U$, there exists a classical function $f=f_U$ such that $R^f$ approximates $U$ (say, to within some distance $\epsilon$ in operator norm). This question is attributed to Aaronson and Kuperberg in 2007 and is still open. A recent work \cite{LMW24} gives evidence that the answer might be negative by showing that there is no such $R$ that makes only a single query to $f$.

\subsection{Our Results}

Our first contribution, in Section~\ref{sec:designstonets}, is a connection between ROM-PRUs and $\epsilon$-nets. An $\epsilon$-net over $\unitary(d)$ (the $d$-dimensional unitary group) is a set of unitaries $\net$ such that for all $U \in \unitary(d)$ there exists $V \in \net$ with $\norm{U-V} \le \epsilon$, where $\| \cdot \|$ denotes the operator norm. We observe that this can be done by thinking of ROM-PRUs as \emph{unitary (approximate) $t$-designs}. A $t$-design is a distribution over unitaries that is functionally identical to Haar-random unitaries for algorithms that make at most $t$ queries. There are various ways to define approximate designs, but for the sake of this exposition we consider the so-called diamond approximation. A $\delta$-approximate $t$-design is one where any adversary making at most $t$ queries has at most $\delta$ statistical advantage in distinguishing the design from random. We call this diamond approximation since it corresponds to the diamond distance between the respective channels. A ROM-PRU with query bound $\qb$ induces (together with the distribution over the random oracle) an approximate $\qb$-design, and therefore we seek a connection between such approximate designs and $\epsilon$-nets.

The question of quantitatively connecting unitary designs to $\epsilon$-nets has only been studied recently~\cite{oszmaniec2021epsilon,S_owik_2025}, but to the best of our knowledge not in the context of diamond approximation. This notion, as we show, provides a more operational handle on the problem of finding a connection with $\epsilon$-nets and in particular is more naturally manipulated using a cryptographic approach, as explained in the technical overview below.

Specifically, we show that a diamond approximate $t$ design, even for very mild $\delta$, e.g.\ $\delta = 0.1$, implies a notion that we call a \emph{relaxed} $\epsilon$-net, so long as $t \gtrsim d^2/\epsilon^2$ (or even $d^2/\epsilon$, if $\delta=0$). A relaxed $\epsilon$-net can be thought of as an average-case property, or as a net that does not cover the entire group $\unitary(d)$, but rather only guarantees that with high probability over a Haar random $U \in \unitary(d)$, there exists a $V \in \net$ s.t.\ $\norm{U-V} \le \epsilon$. It is then simple to show that if we compose a relaxed net with itself (i.e.\ take all matrices of the form $V_1 V_2$, where $V_1, V_2 \in \net$), we get a proper $\epsilon$-net. We provide a detailed technical comparison of our results to the prior works \cite{oszmaniec2021epsilon,S_owik_2025} in Section~\ref{sec:designscompare}.

Translating these results back to ROM-PRUs (Section~\ref{sec:rompru}), we get that if $\qb \gtrsim d^2/\epsilon^2$ (or $d^2/\epsilon$ under a stronger definition than diamond approximation), composing the ROM-PRU with itself implies an $\epsilon$-net. We then observe that an $\epsilon$-net in an oracle model trivially implies unitary synthesis. Specifically, the reduction $R$ is an application of the ROM-PRU (composed with itself), and we have a guarantee that for all $U$ there exists an instantiation of the random oracle to a function $f$ such that $R^f$ is close to $U$. This concludes our first result.

We recall that existing constructions of PRU \cite{MPSY24,MH25,LQSYZ25} are all proven with $\qb \gtrsim \sqrt{d}$ (we note that for ``strong'' PRU where the adversary has access to the inverse, transpose and conjugate of the unitary, existing proofs only show $\qb \gtrsim d^{1/32}$ \cite{MH25,SMLBH25}). This means that ``only'' a polynomial improvement stands in the way of affirmatively resolving the unitary synthesis problem, contrary to the evidence in \cite{LMW24}. Indeed, it may be the case that such a polynomial improvement is impossible, which would imply that some unknown consideration prevents us from achieving stronger unitary pseudorandomness. We believe that exploring this gap would shed light on the nature of this important class of primitives.

Given that the gap between what is known and what is sought is only polynomial, it makes sense to check whether the $\sqrt{d}$ bound in existing constructions is tight, or whether existing constructions come even closer to the $d^2$ tipping point. We consider the now-famous PFC construction of PRUs, introduced in \cite{MPSY24} and proven adaptively secure in \cite{MH25}. We show in Section~\ref{sec:PFC} that for this construction, it is possible to violate security using roughly $\sqrt{d}$ queries; furthermore, it suffices to query the unitary on the same input $\ket{0}$ in all queries. This result sharpens the motivation to find \emph{any} construction that breaks the $\sqrt{d}$ barrier. We note in this context that there are many ways to tweak the construction so that our techniques no longer work as is, even just sequential repetition. It remains an open problem whether there is a $\sqrt{d}$ upper bound for more general classes of constructions.

Lastly, in Section~\ref{sec:entropylb}, we analyze existing and prospective constructions of ROM-PRUs from an entropic standpoint. In particular, we show that relatively simple bounds impose strong constraints on what is achievable using standard techniques. We start by recalling some bounds from the literature on the cardinality of $t$-designs. We plug them into a simple counting argument for a ROM-PRU that uses an oracle $f: \binset^m \to \binset$. The support size of such a construction is therefore at most $2^{2^m}$. It follows that such a construction can only implement an approximate $t$-design if $m \gtrsim \log t + \loglog(d^2/t)$. Since in all known constructions it is only known how to use $m \le \log d$, we get a bound $t \lesssim d$. Indeed, such constructions cannot be scalable or even approach the $t=d^2$ tipping point. Note the intriguing gap between this bound and the actual $t = \sqrt{d}$ result that can be proven.

We notice that the aforementioned bound only holds for $t \ll d^2$ (for larger values $\loglog (d^2/t)$ is undefined). Indeed, we are not aware of any lower bound in the literature that applies for $t \gg d^2$. We therefore proceed to develop such a lower bound. We provide two proofs of similar statements: one from first principles using representation theory, and the other almost immediately from our previous results, using the connection to relaxed $\epsilon$-nets that we described above. We find it quite satisfactory that this bound essentially matches the best known bound, using a very different approach.

We proceed to show that approaching $t=d^2$ requires techniques different from those in the literature. Specifically, we show that relaxed $\epsilon$-nets in an oracle model require the oracle's \emph{input} to have length at least $2 \log d - o(\log d)$. That is, in order to achieve unitary synthesis for $n$ qubits, it is essential to use classical functions with input length roughly $2n$. Such a claim follows immediately from a counting argument if we consider only classical functions with binary output. However, our results apply even in a generalized model that allows classical functions with \emph{infinite} output length. Indeed, quantum access to a binary classical function is equivalent to access to a \emph{diagonal} unitary with $\pm 1$ values (and whose dimension is $2^n \times 2^n$). Likewise, if the range of the function is $[q]$, then a $q$-ary root of unity replaces $(-1)$. We may therefore allow arbitrary diagonal unitary gates, which would allow infinite precision in the description of the roots of unity along the diagonal, and potentially provide the necessary degrees of freedom required to construct ROM-PRUs and unitary synthesis. We show that this hope cannot be realized: even in the diagonal matrix model, it is impossible to achieve relaxed nets, and therefore also unitary synthesis and ROM-PRUs with $\qb \gg d^2$.

This result in particular applies to all existing candidates for PRUs. No known polynomial-size construction manages to utilize oracles of large input length. The state-of-the-art unitary synthesis algorithm due to Rosenthal~\cite{R21} requires the circuit $R$ to have size $\sqrt{d}$, and it indeed uses an oracle with input length $2\log d$. Viewed from our perspective, the $\sqrt{d}$ queries in that work are required to uncompute half of the $2 \log d$-qubit register provided to the oracle, in order to recover an output of length $\log d$. It is quite curious that there are no known techniques for nontrivially utilizing classical functions of input length greater than $\log(d)$, and we highlight this as a challenge for future ROM-PRU constructions.

\subsection{Technical Overview}

We now provide an overview of the techniques used in the paper.

\paragraph{Connections between ROM-PRU, Unitary Designs and Nets.}
In Section~\ref{sec:designstonets} we show connections between diamond designs and $\epsilon$-nets. As explained above, working with diamond designs allows us to take a cryptographic approach to relating them to $\epsilon$-nets. Indeed, a $\delta$-approximate $t$-design in the diamond norm is simply a distribution that is $\delta$-statistically indistinguishable from Haar random using $t$ nonadaptive queries. Our intuition, therefore, is to show that if our design is not an $\epsilon$-net, then it can be efficiently distinguished from the uniform distribution. A straightforward test that a set $\net$ is not an $\epsilon$-net is to sample a random unitary $U$ and calculate its distance to all elements in $\net$; if the unitary is $\epsilon$-far from all elements in $\net$, then $\net$ is not an $\epsilon$-net. Likewise, if the probability of being far from all elements in $\net$ is high, then we do not even have a relaxed $\epsilon$-net, as we explained above. Since we are only given \emph{query access} to $U$, our strategy is to first \emph{learn} $U$ and obtain a unitary $\hat{U}$ that has some guarantee of resembling $U$, and then compare $\hat{U}$ against all elements of $\net$, which does not require any additional queries.

We use the learning algorithm of \cite{haah2023query}, which requires $\approx d^2 \log (1/\eta)/\epsilon$ adaptive queries or $d^2 \log (1/\eta)/\epsilon^2$ non-adaptive queries to learn $U$ to within distance $\epsilon$ with probability at least $1-\eta$. We consider a distinguisher that makes the necessary number of queries to obtain $\hat{U}$, compares it against all elements in $\net$, and returns $1$ if and only if there is no element in $\net$ that is $\epsilon$-close to $\hat{U}$. When drawing $U$ from the approximate $t$-design, the probability of output $1$ is at most $\eta$, essentially by definition. This means, using the indistinguishability property, that the probability of returning $1$ on a \emph{random} $U$ is at most $\eta + \delta$. It follows that the support $\net$ induced by our approximate design is also a relaxed $\epsilon$-net, with a fraction of at most $\eta + \delta$ of the unitary group not being covered.

To go from a relaxed net to an $\epsilon$-net we use composition. Indeed, this can be seen as a version of the worst-case to average-case reduction of the unitary synthesis question. Let $\net$ be a relaxed net that covers most of the unitary group, i.e.\ with probability $> 0.5$, for a random $U$, there exists a $V$ that is $\epsilon$-close to $U$. Now consider an arbitrary unitary $T$ and the two dependent random variables $U^\dagger$ and $UT$, where $U$ is Haar random. Since both marginals are Haar random, by the union bound there is a positive probability that both are within $\epsilon$ distance from $\net$. Therefore, there exist $V_1, V_2 \in \net$ such that $\norm{V_1 - U^\dagger} \le \epsilon$ and $\norm{V_2 - UT} \le \epsilon$. It follows from submultiplicativity that $\norm{V_1 V_2 - T} \le 2\epsilon$. Therefore, $\net \cdot \net$ is a full $2\epsilon$-net.

Our results establish a direct relation between diamond-designs and $\epsilon$-nets. This is different in terms of techniques and incomparable in terms of the final result from prior work that relied in the notion of ``tensor-product extractors''. We provide a comparison and discussion in Section~\ref{sec:designscompare}.

Section~\ref{sec:rompru} focuses on the notion of ROM-PRU, and in particular its scalable version. We provide a definition of scalable ROM-PRU and use the aforementioned results to show that a ROM-PRU with a sufficiently large $\qb$ would affirmatively resolve the unitary synthesis problem.

\paragraph{An Upper Bound for Existing Constructions.} In Section~\ref{sec:PFC}, we show that using roughly $\sqrt{d}$ queries to the PFC construction, we can distinguish it from uniform. Furthermore, this is done by simply querying the PFC construction on the fixed input $0$, and measuring the resulting state in the computational basis. The property used for distinguishing is the number of collisions in this experiment. That is, if we make $t$ queries, how many of the $\binom{t}{2}$ pairs of strings have the same value?

To understand why this test could be useful, let us zoom into the specifics of the PFC construction. As the name suggests, this construction contains $3$ components. First, the input runs through a random Clifford $C$. Then it runs through a random phase gate $\ket{x} \to \omega^{f(x)} \ket{x}$, where $\omega$ is a root of unity (either $(-1)$ as in \cite{MPSY24}, or a higher order root), and $f$ is a random classical function with the proper domain and range. Finally, it runs through a random permutation $P$. The crucial observation is that if we are just looking at collisions on computational basis elements, then $F$ and $P$ are irrelevant, since they do not affect whether a collision occurs or with what probability. Therefore, for the sake of the collision-counting problem, the PFC construction is equivalent to a random Clifford, which at least intuitively should be quite different from a Haar-random unitary.

For a Haar random unitary, the state $U \ket{0}$ is a Haar random state, and its statistical properties are very well known. In particular, the probability of collision between two measurements of copies of the same string is roughly $2/d$, and this value is extremely well concentrated. Namely, with probability close to $1$ over $U$, we expect $\frac{2}{d} \binom{t}{2}$ collisions. For $t \approx \sqrt{d}$, this becomes a noticeable quantity.

For a Clifford, however, the state $C \ket{0}$ is a random \emph{stabilizer} state. Measuring a stabilizer state in the standard basis induces a uniform distribution over an affine $\mathbb{F}_2$ subspace of dimension $k$, where $k$ depends on the specific Clifford $C$, i.e.\ a collision probability of $2^{-k}$. Indeed, we show that the event $k = \log(d)$ has constant probability. Therefore, with constant probability there is a constant gap between the number of collisions expected in the two experiments.

In order to have a statistically significant test, we use the Median of Means estimator and show (by analyzing the variance of the collision number) that by performing the test correctly we have a constant distinguishing gap between Haar and PFC.

\paragraph{Entropy-Based Bounds for Designs and Unitary Synthesis.} We show our lower bounds in Section~\ref{sec:entropylb}. As explained above, we derive a new lower bound for the cardinality of the support of an approximate $t$-design, for $t \gg d^2$ in two ways. The first relies on representation theory and we will not survey it here. The second is quite straightforward given our work. If $t \gg d^2$, then a unitary $t$-design is also a relaxed $\epsilon$-net, i.e.\ it $\epsilon$-covers some $1-\eta$ fraction of the volume of the unitary group. Showing a lower bound on the cardinality of a relaxed net is trivial using volume considerations: since the volume of an $\epsilon$-ball is roughly $\epsilon^{d^2}$, the cardinality of a relaxed design is at least $(1-\eta)\left(\frac{1}{\epsilon}\right)^{d^2}$. Our transformation requires $t \approx d^2/\epsilon^2$ for diamond designs, which implies a bound of roughly $(1-\eta)\left(\frac{t}{d^2}\right)^{d^2/2}$, where $\eta$ is related to $\delta$. This turns out to be very similar (up to some constants and low order terms) to the bound that we get from representation theoretic analysis.

As explained above, unitary synthesis is equivalent to having an $\epsilon$-net in the classical oracle model. Therefore, the cardinality bound for $\epsilon$-nets immediately implies that unitary synthesis with a binary function requires an input length of roughly $2 \log d$.

Next, for our diagonal matrix result we again employ a simple principle. Any diagonal matrix in $\unitary(d)$ can be approximated by a diagonal matrix with elements of the form $\omega^\ell_K$, where $\omega_K = e^{\frac{2 \pi i}{K}}$ and $\ell \in [K]$. Note that $\ell$ can be represented by $k = \log K$ bits. It can be shown that the quality of this approximation scales with $\frac{1}{K}$. If we consider a candidate for unitary synthesis, then the number of calls to the diagonal oracle is bounded by the circuit size of the reduction $R$, which is polynomial (we can handle super-polynomial reductions as well); denote its size by $s$. Then so long as $\frac{s}{K}$ is sufficiently small, the outcome of running $R$ with a diagonal oracle versus a $K$-truncated diagonal oracle should not be noticeable. For concreteness let us take $k = \polylog d$ so that $K = d^{\polylog d}$. Then even a huge reduction size cannot notice the truncation. However, the truncated version corresponds to oracle access to a classical function with $\polylog d+O(1)$ \emph{output} bits, which is equivalent to having a binary function whose input size increases \emph{additively} by $\loglog\, d+O(1)$. Therefore, even working with diagonal matrices of infinite precision cuts only an additive $\loglog \,d$ factor from the required input length, compared to the binary-function setting. For example, for unitary synthesis, even in the diagonal matrix model, we require input length of at least $2 \log d - \loglog d - O(1)$.

\section{Preliminaries}
\label{sec:prelims}

Throughout the paper we use approximate equalities (``$\approx$'') and inequalities (``$\lesssim$'' and ``$\gtrsim$'') to indicate relationships between two quantities without specifying universal constants or potentially uninteresting logarithmic factors. For example, for variables $t,d$ we might write $t \gtrsim \frac{d^2}{\eps^2}$ to mean that there are universal constants $c > 0$ and $C \geq 1$ such that $t \geq c \frac{d^2}{\eps^2} + C$.

We let $\unitary(d)$ denote the unitary group over $\C^d$, and $\id_d$ denote the identity element. We slightly abuse notation and let $U \in \unitary(d)$ also denote the quantum channel $U (\cdot) U^\dagger$ over $\C^d$ (which is a superoperator over $\C^{d \times d}$). We let $\hr(d)$ denote the Haar measure over $\unitary(d)$; we sometimes omit the $d$ when it is clear from context. A distribution of unitaries $\udesign$ is \emph{symmetric} if $\udesign^\dagger = \udesign$; that is, sampling a unitary $U \sim \udesign$ and taking its Hermitian conjugate $U^\dagger$ yields the same distribution $\udesign$.

For a superoperator $T$ we consider the following norms.
\begin{definition}
	For a superoperator $T$, the $k\to k$ norm is defined as
	\[
	\| T \|_{k \to k} = \sup_{X \neq 0} \frac{ \| T(X) \|_k}{ \| X \|_k}
	\]
	where $\|X\|_k$ denotes the Schatten-$k$ norm (for $k=2$ this is the Frobenius norm) of an operator $X$.
\end{definition}
Some important special cases include the $2\to 2$ norm (expander norm) and the $1\to 1$ norm. An important derived norm is the completely bounded $1\to 1$ norm, also known as the diamond norm defined as
\[
\| T \|_{\diamond} = \sup_{m \in \N}  \| T \otimes \id_m \|_{1 \to 1}~.
\]

\subsection{Nets}

For any set of unitaries $\net$, we let $\net^{\dagger}$ denote $\{V : V^{\dagger} \in \net\}$. If $\net_1, \net_2$ are two sets then $\net_1 \cdot \net_2 = \{V_1\cdot V_2 : V_1 \in \net_1, V_2\in\net_2\}$. We sometimes denote singleton sets $\{\unitary\}$ simply by $\unitary$.

We now define a metric to quantify how much a set of unitaries ``covers'' the entire unitary group. We say that a unitary is covered by the set $\net$ if it is sufficiently close to $\net$. The usual definition of $\eps$-nets requires that all unitaries are $\eps$-close to $\net$. For our purposes it is important to consider a refinement that we call an $(\eps, \eta)$-net, where all but $\eta$-measure of the unitary group is $\eps$-covered. A formal definition follows.
\begin{definition}
	The $\eps$-\emph{coverage} of a set of unitaries $\net \subseteq \unitary(d)$ is defined as
	\[
	\cvg(\net) = \cvg_\eps(\net) = \left\{ U \in \unitary(d) :   \exists \, V \in \net \text{ such that } \| U( \cdot ) U^\dagger - V ( \cdot )V^\dagger \|_\diamond \leq \eps \right\}~.
	\]
	The $\epsilon$-\emph{exposure} of $\net$ is defined as the measure of unitaries that are not covered:
	\[
	\xps(\net) = \xps_\eps(\net) = \Pr_{U \sim \hr(d)} \Big [ U \not\in \cvg(\net) \Big]~.
	\]
	A set $\net$ is an $(\eps, \eta)$-net if $\xps_\eps(\net) \le \eta$.
	We sometimes denote $\vol(\net) = \vol_\eps(\net) = \hr(\cvg_\eps(\net))$. Note that $\vol_\eps(\net) = 1- \xps_\eps(\net)$.
\end{definition}

The notion of $(\eps, 0)$-net coincides with the usual notion from the literature of $\epsilon$-nets.\footnote{This is not completely straightforward from the definition, since on the face of it, there could exist an exposed measure $0$ set of unitaries that is not empty. However, since any exposed set is necessarily an open set, such a degeneracy is not possible.}
We note that one can use a different distance measure for defining $\eps$-closeness, however most relevant distance metrics would differ by a $\poly(d)$ factor, which for our purposes is usually not significant.
We proceed by providing a counting argument for the minimal cardinality of a net.

\begin{proposition}[Fact 2 in \cite{oszmaniec2021epsilon}]
	There exists a universal constant $c_\diamond > 0$ such that the Haar volume of an $\epsilon$-ball in diamond norm over $d$ dimensional unitary channels is at most $(\epsilon/ c_\diamond)^{d^2-1}$.
\end{proposition}
The following is an immediate consequence.
\begin{corollary}
	\label{cor:net-size}
	Let $\net$ be an $(\eps, \eta)$-net. Then
	\[
	\Big | \net \Big| \geq (1-\eta)\Big(  \frac{c_\diamond}{\eps} \Big)^{d^2-1}~.
	\]
\end{corollary}

\subsection{Unitary Designs}

A unitary (approximate) $t$-design is a distribution $\udesign$ over unitary matrices that is similar to the Haar random measure so long as the associated channel is only invoked $t$ times. There are various definitions for this object. In this work, we almost always use the notion of ``diamond-design'' (see below), but we sometimes compare with other notions in order to put our results in context.

\begin{definition}[Approximate unitary designs]
	Let $\udesign$ be a distribution over $\unitary(d)$ and define the moment operator
	\[
	\Phi_{\udesign,t}(\cdot) = \E_{U \sim \udesign} \left[ U^{\otimes t} (\cdot) U^{\otimes t,\dagger} \right]~.
	\]
	We say that $\udesign$ is: %
	\begin{enumerate}
		\item A $(t,\delta)$-diamond-design if:
		\[
		\| \Phi_{\udesign,t} - \Phi_{\hr,t} \|_\diamond \leq \delta~.
		\]
		\item A $(t,\delta)$-TPE (tensor-product expander) if: %
		\[
		\| \Phi_{\udesign,t} - \Phi_{\hr,t} \|_{2 \to 2} \leq \delta~.
		\]
		\item A $(t,\delta)$-relative-design if:
		\[
		(1 - \delta) \Phi_{\hr,t} \leq \Phi_{\udesign,t} \leq (1 + \delta) \Phi_{\hr,t}~.
		\]
	\end{enumerate}
\end{definition}

Some relations between the different notions are provided below.
\begin{proposition}
\label{prop:norm-conversions}
	Let $\udesign$ be a distribution and $t$ be some integer. Let $\lambda, \delta, \epsilon$ be such that $\udesign$ is a $(t,\delta)$-diamond-design, $(t,\lambda)$-TPE and $(t,\epsilon)$-relative-design. Then the following hold.
	\begin{enumerate}
		\item $\delta \le \epsilon$
		\item $\lambda \le d^{t/2}\delta$. %
		\item $\delta \le d^t \lambda$
		\item $\epsilon \le d^{2t} \lambda$
	\end{enumerate}
\end{proposition}
The proofs for the different relations appear in \cite{low2010pseudo,brandao2016local,Mele2024introductiontohaar,CHHLMT24}.

The following proposition can be derived as a corollary of~\cite[Corollary 2.11]{CHHLMT24}; we provide an elementary proof for completeness.
\begin{proposition}
	If $\udesign$ is a symmetric (meaning $\udesign = \udesign^\dagger$) $(t,\delta)$-diamond design, then it is also a $(t,\delta)$-TPE.
\end{proposition}
\begin{proof}
	Let $m$ be some integer, and consider $\udesign_m$ as the $m$-composition of $\udesign$. Namely, each entry in $\udesign_m$ is a sequential composition of $m$ i.i.d elements from $\udesign$. Let $\lambda_m, \delta_m$ denote its $t$-design approximation with respect to the $2\to 2$ and diamond norms respectively. Then it holds that $\lambda_m \le d^{t/2} \delta_m$ by the definitions of the norms (see e.g.\ \cite[Proposition~42]{Mele2024introductiontohaar}). Furthermore, if $\udesign$ is symmetric then it holds that $\lambda_m = \lambda^m$ (see e.g.\ \cite[Footnote~4]{Mele2024introductiontohaar}). Lastly, it holds that $\delta_m \le \delta^m$ \cite[Lemma 4.11]{MPSY24}. We conclude that for all $m$, $\lambda^m = \lambda_m \le d^{t/2} \delta_m \le d^{t/2} \delta^m$. That is, for all $m$, $\lambda \le d^{t/(2m)} \delta$, which implies that $\lambda \le \delta$.
\end{proof}

Finally, the following gives operational consequences of the definitions of diamond and relative error designs.
\begin{proposition}\label{prop:designdistinguish}
	The following relations exist between designs and distinguishability via quantum query algorithms.
	\begin{enumerate}
		\item 	A distribution $\udesign$ is a $(t,\delta)$-diamond-design if and only if for any \emph{non-adaptive} $t$-query oracle algorithm $\cA$ with $\{0,1\}$ output it holds that
		\[
		\abs{\E_{U \sim \udesign}[\cA^{U}] - \E_{U \sim \hr}[\cA^{U}]} \le \delta~.
		\]

		\item A distribution $\udesign$ is a $(t,\delta)$-relative-design only if for any (possibly adaptive) $t$-query oracle algorithm $\cA$ with $\{0,1\}$ output it holds that
		\[
		\abs{\E_{U \sim \udesign}[\cA^{U}] - \E_{U \sim \hr}[\cA^{U}]} \le \delta~.
		\]
        Here, the algorithm $\cA$ can only query controlled-$U$ (but not its inverse).
		Note that there is no equivalence between relative-designs and indistinguishability by algorithms here, only a forward implication.
	\end{enumerate}
\end{proposition}
\begin{proof}
	The first item follows directly by definition. For the second item see~\cite[Lemma 25]{kretschmer2021quantum}.
\end{proof}

\section{From Designs to Nets}
\label{sec:designstonets}

In this section, we present a new method for converting (approximate) designs into $\epsilon$-nets. Our approach differs from previous works in two main respects. First, whereas previous works relied on TPE properties, we use the diamond norm or relative distance. Indeed, we show that even an extremely mild approximation under these metrics implies $\epsilon$-nets. Second, we do not seek to show that the support of the approximate design itself constitutes an $\epsilon$-net, as in previous work. Instead, we allow simple transformations of the approximate design, such as squaring or closing the set under inverse. In terms of techniques, we show a very elementary connection between approximate designs and $(\epsilon,\eta)$-nets, and then show how (standard) $\eps$-nets follow in a straightforward manner.

A central building block for our construction is the family of channel tomography procedures from Haah et al.\ \cite{haah2023query}.

\begin{proposition}[{Unitary Tomography~\cite{haah2023query}}]\label{prop:channeltom} There are procedures $T_{i, \epsilon, \eta, d}$ that make oracle queries to an unknown unitary $U \in \unitary(d)$ and produce a classical description of a unitary $\hat{U}$ such that $\| U (\cdot) U^\dagger - \hat{U} (\cdot) \hat{U}^\dagger \|_\diamond \leq \eps$ with probability at least $1-\eta$, with the following parameters:
	\begin{enumerate}
		\item $T_{1,\epsilon, \eta, d}$ makes $t = \Theta(\frac{d^2}{\eps^2} \ln \frac{1}{\eta})$ non-adaptive queries to $U$.
		\item $T_{2,\epsilon, \eta, d}$ makes $t = \Theta(\frac{d^2}{\eps} \ln \frac{1}{\eta})$ adaptive queries to $U$.
	\end{enumerate}
\end{proposition}
We sometimes remove some of the subscripts when they are clear from the context.

We can now state and prove our first technical lemma, which converts approximate designs into relaxed nets.

\begin{lemma}[Designs to nets]
	\label{lem:designs-to-nets}
	Let $\udesign$ be a distribution over $\unitary(d)$ and $\net = \support(\udesign)$. Let $\eps \in (0,1)$. Then if $\udesign$ is a $(t,\delta)$-diamond-design (respectively relative-design), for $t = \Theta(\frac{d^2}{\eps^2}\ln \frac{1}{\eta_0})$ (respectively $t = \Theta(\frac{d^2}{\eps}\ln \frac{1}{\eta_0})$) then $\net$ is an $(\eps,\eta)$-net, for $\eta = \frac{\delta+\eta_0}{1-\eta_0}$. In particular if $\eta_0 \le 1/6$ and $\delta \le 1/6$ then $\eta \le 2/5$.
\end{lemma}
\begin{proof}
	We use a tomography argument. Suppose $\net$ were not an $(\eps,\eta)$-net. Then there must exist a set $S \subseteq \unitary(d)$ of measure at least $\eta$ such that every $V \in S$ is at least $\eps$-far, in diamond distance, from every unitary in $\net$.

	Consider the following statistical test $\cA$ to distinguish $\udesign$ from the Haar measure, set $i=1$ for the diamond-design setting and $i=2$ for the relative-design setting.
	\begin{enumerate}
		\item Run the algorithm $T^U_{i, \epsilon/3, \eta_0, d}$ from Proposition~\ref{prop:channeltom} to obtain $\hat{U}$.

		\item If $\hat{U}$ is within $\eps/3$ of a point in $S$, then output $1$. Otherwise, output $0$.
	\end{enumerate}
	By Proposition~\ref{prop:channeltom} and Proposition~\ref{prop:designdistinguish}, it follows that in both settings of this lemma statement, $\cA$ is a distinguisher for the $(t,\delta)$ design, according to the respective metric. Furthermore, the invocation of $T$ within $\cA$ outputs $\hat{U}$ at distance at most $\epsilon/3$ from its input $U$ with probability at least $1-\eta_0$.

	We now analyze the acceptance probability of $\cA$ on the distributions $\hr$ and $\udesign$.

	When $U$ is drawn from the Haar measure, then
	\begin{align}
		\E_{U \sim \hr} [\cA^U] \geq (1 - \eta_0) \eta~.
	\end{align}
	This is the probability of the tomography algorithm succeeding, and conditioned on it succeeding, the probability that it samples a unitary $U \in S$ which is at least $\eta$. By definition the estimated unitary $\hat{U}$ is at most $\eps/3$-far from $S$.

	On the other hand, when $U$ is sampled from $\udesign$, then
	\begin{align}
		\E_{U \sim \udesign} [\cA^U] \leq \eta_0
	\end{align}
	because if the tomography procedure succeeds, the estimated unitary $\hat{U}$ will be within $\eps/3$ of $U \in \net$, which by definition is at least $(2\eps/3)$-far from $S$. Thus the procedure $\cA$ can only accept when the tomography algorithm fails.

	Thus the statistical test can distinguish between $\udesign$ and the Haar measure with advantage at least $(1 - \eta_0)\eta - \eta_0$. We get a contradiction if this value is larger than $\delta$. Namely, we have that $\net$ is an $(\epsilon, \eta)$-net with $\eta = \frac{\delta + \eta_0}{1-\eta_0}$. In particular, if $\eta_0 < 1/6$ and $\delta < 1/6$, then $\eta < 2/5$.
\end{proof}

Up to this point, we have shown that the support of a design constitutes a relaxed $\epsilon$-net. We now show an extremely simple composition lemma that allows us to go from a relaxed net to a full net. In particular, we show how to compose two relaxed nets into a full net.
\begin{lemma}\label{lem:netcompose}
	Let $\net_i$ be an $(\eps_i,\eta_i)$-relaxed net for $i=1,2$. If $\eta_1 + \eta_2 < 1$, then $\net_1 \cdot \net_2^{\dagger}$ is an $(\epsilon_1+\epsilon_2)$-net.
\end{lemma}
\begin{proof}
	Let $U$ be a unitary, and consider the nets $\net_1$ and $U \cdot \net_2$. It holds that $\xps_{\eps_1}(\net_1) + \xps_{\eps_2}(U \net_2) = \eta_1 + \eta_2 < 1$. Therefore, $\cvg_{\eps_1}(\net_1) \cap \cvg_{\eps_2}(U \net_2)$ is nonempty. Let $T$ be a unitary in this intersection. Then by definition there exist $V_1 \in \net_1$, $V_2 \in \net_2$ such that both $\| T( \cdot ) T^\dagger - V_1 ( \cdot )V_1^\dagger \|_\diamond \leq \eps_1$ and $\| T( \cdot ) T^\dagger - (U V_2) ( \cdot )(UV_2)^\dagger \|_\diamond \leq \eps_2$. By triangle inequality we have $\| V_1( \cdot ) V_1^\dagger - (U V_2) ( \cdot )(UV_2)^\dagger \|_\diamond \leq \eps_1 + \eps_2$. By unitary invariance of the metric, we have $\| (V_1V_2^{\dagger})( \cdot ) (V_1V_2^{\dagger})^\dagger - U  ( \cdot )U^\dagger \|_\diamond \leq \eps_1 + \eps_2$. It follows that $U \in \cvg_{\eps_1+\eps_2}(\net_1 \net_2^{\dagger})$ and the lemma follows.
\end{proof}

Using the above, if we have a good enough net, we can compose it with its inverse, or even with itself.
\begin{proposition}
	If $\net$ is an $(\eps, \eta)$-relaxed net, then $\net^{\dagger}$ is also $(\eps, \eta)$-relaxed.
\end{proposition}
\begin{proof}
	Consider the following derivation
	\begin{align*}
		\Pr_{U \sim \hr}[U \not\in \cvg_\eps(\net^\dagger)] & = \Pr_{U \sim \hr}[U^\dagger \not\in \cvg_\eps(\net)]\\
		& = \Pr_{U \sim \hr}[U \not\in \cvg_\eps(\net)]
	\end{align*}
	where the first equality follows from the inverse-symmetry of the diamond norm, and the second equality follows from the symmetry of the Haar measure.
\end{proof}

We get the following corollaries as a direct implication of Lemma~\ref{lem:netcompose}.
\begin{corollary}
	\label{cor:relaxed-net-to-net}
	Let $\net$ denote a $(\eps,\eta)$-net, $\eta < 1/2$. Then $\net^2$ and $\net \net^{\dagger}$ are $2\eps$-nets.
\end{corollary}

\subsection{Comparison with Prior Work} \label{sec:designscompare}

The quantitative connection between unitary designs and nets was first studied in~\cite{oszmaniec2021epsilon}, and then improved upon by~\cite{S_owik_2025}. The main result of~\cite{S_owik_2025} is the following:

\begin{theorem}[TPE implies nets~\cite{oszmaniec2021epsilon,S_owik_2025}]
\label{thm:slowik-tpe-implies-nets}
Let $\nu$ be a $(t,\lambda)$-TPE with
\[
    t \gtrsim \frac{d^{5/2}}{\eps} \, , \qquad \qquad \lambda \lesssim \Big( \frac{\eps}{d^{1/2}} \Big)^{d^2 - 1}
\]
then $\support(\nu)$ is an $\eps$-net. Here, ``$\gtrsim$'' and ``$\lesssim$'' hide log-factors in $d$ and $\eps^{-1}$ as well as constants.
\end{theorem}

We point out several important differences. First, the hypothesis of Theorem~\ref{thm:slowik-tpe-implies-nets} assumes a tensor-product expander (TPE), rather than a diamond-design as in our result Lemma~\ref{lem:designs-to-nets}. Using Theorem~\ref{thm:slowik-tpe-implies-nets} to relate diamond designs to nets would yield the following: a symmetric\footnote{We recall that a distribution $\nu$ is symmetric if $\nu = \nu^\dagger$.} $(t,\delta)$-diamond-design yields a (relaxed) $\eps$-net only if $t \gtrsim d^{5/2} \eps^{-1}$ and $\delta \lesssim (\eps d^{-1/2})^{d^2 - 1}$, which is (in some ways) weaker than our Lemma~\ref{lem:designs-to-nets}, which handles any $\delta$ and requires $t \gtrsim d^2 \eps^{-2}$ (i.e., a better dependence on $d$, but a worse dependence on $\eps$).

Second, our result relating diamond-designs to nets can be used to give a relatively simple proof of a statement very similar to Theorem~\ref{thm:slowik-tpe-implies-nets}.
\begin{corollary}
Let $\udesign$ be a $(t,\lambda)$-TPE with
\[
    t \gtrsim \frac{d^2}{\eps^2} \ln \frac{1}{\eta_0} \, , \qquad \qquad \lambda \leq d^{-t} \delta
\]
for some $\eta_0,\delta > 0$, then $\support(\udesign)$ is a $(\eps,\eta)$-net for $\eta = \frac{\delta + \eta_0}{1 - \eta_0}$.
\end{corollary}
\begin{proof}
    If $\udesign$ is a $(t,\lambda)$-TPE, then by Proposition~\ref{prop:norm-conversions}, it is also a $(t,d^t \lambda)$-diamond-design. The corollary statement follows directly from Lemma~\ref{lem:designs-to-nets}.
\end{proof}
In this corollary, the dependence on $d$ in the lower bound for $t$ is better, but the dependence on $\eps$ in the upper bound on $\lambda$ is worse ($\eps$ appears in the exponent of $d^{-t}$). The result is also a relaxed net rather than a bona fide net (although Corollary~\ref{cor:relaxed-net-to-net} shows how this can be used to build an $\eps$-net).

Finally, we observe that the techniques used to prove Lemma~\ref{lem:designs-to-nets} are very different from those used to prove Theorem~\ref{thm:slowik-tpe-implies-nets} in~\cite{oszmaniec2021epsilon,S_owik_2025}, which were highly analytical and geometrical. Our proof is arguably much more intuitive and operational (e.g., using channel tomography as part of the argument). It is an interesting question whether our techniques can be used to completely recover the statement of Theorem~\ref{thm:slowik-tpe-implies-nets}.

\section{ROM-PRUs and Scalability}
\label{sec:rompru}

A pseudorandom unitary (PRU) ensemble~\cite{JLS18} is a distribution over efficiently implementable unitaries that cannot be distinguished from Haar-random by a computationally efficient distinguisher that makes few queries to the unitary. The security of PRUs (as defined) requires a cryptographic hardness assumption, such as the existence of post-quantum one-way functions.

All prior PRUs with security based on standard cryptographic assumptions are analyzed in two steps: first, the PRU is implemented as an algorithm that has access to a random Boolean function, and shown to be information-theoretically secure against a distinguisher that makes few queries. Notably, the distinguisher is not required to be computationally efficient. Then, the PRU implementation is made efficient by replacing the random Boolean function with the output of a post-quantum pseudorandom function (PRF), and consequently the PRU is secure against computationally bounded distinguishers.

We formalize the object constructed in the first part of these analyses; we call it a \emph{ROM-PRU}.

\begin{definition}[ROM-PRU]
\label{def:rompru}
    A family of $d$-dimensional unitaries $\mathscr{U} = \{ U_{f} \}_{f}$ indexed by boolean functions $f:\{0,1\}^{m} \to \{0,1\}$ is a ROM-PRU with
    \begin{itemize}
        \item Dimension $d$,
        \item Construction query complexity $q$,
        \item ROM input length $m$,
        \item Implementation error $\alpha$,
        \item Security against $t$-query adversaries with distinguishing advantage $\delta$,
    \end{itemize}
    if and only if
    \begin{enumerate}
        \item \emph{(Efficient computability)}. There exists a $q$-query algorithm $\cA$ such that for all $f:\{0,1\}^{m} \to \{0,1\}$,
    \[
        \| \cA^f(\cdot) - U_{f}\, (\cdot) \, U_{f}^\dagger \|_\diamond \leq \alpha~.
    \]
        \item \emph{(Pseudorandomness)}. For all $t$-query distinguishers $\cD$,
    \[
        \Big | \Pr_f[\cD^{U_{f}} = 1] - \Pr_{U \sim \hr(d)} [\cD^U = 1] \Big | \leq \delta~.
    \]
    Here, the first probability is over a uniformly random function $f:\{0,1\}^{m} \to \{0,1\}$.
    \end{enumerate}
\end{definition}

We compare this with the typical definition of PRU from~\cite{JLS18,MPSY24,MH25} and note several key differences.
\begin{enumerate}
    \item The algorithms $\cA$ (implementing the PRU) and $\cD$ (distinguishing the PRU from Haar-random) are only required to be query-bounded, not computationally bounded.

    \item The algorithm $\cA$ implementing the PRU queries a random boolean function $f:\{0,1\}^{m} \to \{0,1\}$, which can be thought of as analogous to the ``key'' in the standard definition of PRU from~\cite{JLS18}.

    \item The PRUs in~\cite{MPSY24,chen2024efficient,MH25,LQSYZ25,SMLBH25} are assumed to be implemented by quantum algorithms with zero error, whereas in our definition of ROM-PRU we allow for some implementation error.

\end{enumerate}

\begin{remark}
    In our definition of ROM-PRU,
the distinguisher $\cD$ only has query access to the unitary $U_{f}$, and not directly to the random function $f$. This is slightly different from the standard use of the ROM terminology in cryptography. In the context of PRUs, this seemingly weaker security definition is sufficient when the goal is to instantiate the random oracle $f$ using a PRF.
\end{remark}

\begin{remark}
    The reader may wonder about the purpose of the implementation error parameter $\alpha$; after all, couldn't it be folded into the distinguishing advantage? This parameter is used to capture cases in which the implementing algorithm $\cA$ does not implement a unitary exactly, but rather some channel (for example, it might weakly entangle the input with some ancillas). %
\end{remark}

Note that in our definition, a ROM-PRU is a finite set of unitaries indexed by Boolean functions, all of the same dimension. We can consider an asymptotic family of ROM-PRUs, indexed both by dimension and by a security parameter, and define what it means for such a family to be scalable.

\begin{definition}[Scalable ROM-PRU]
\label{def:scalable-rom-pru}
A family $\{ \mathscr{U}_{d,\secp} \}_{d,\secp \in \N}$ of uniformly-generated ROM-PRUs is a \emph{scalable ROM-PRU family} if each $\mathscr{U}_{d,\secp} = \{ U_{d,\secp,f} \}_{f}$ is a $d$-dimensional ROM-PRU with
\begin{enumerate}
    \item Construction query complexity and ROM input length $\poly(\log d,\secp)$,
    \item Implementation error $2^{-\secp}$,
    \item Security against $2^{\secp}$-query adversaries with distinguishing advantage at most $2^{-\secp}$.
\end{enumerate}
\end{definition}

Here, ``uniformly-generated'' means that there is a single uniform quantum algorithm $\cA$ that, when it gets parameters $d,\secp$ as input, implements the ROM-PRU $\mathscr{U}_{d,\secp}$. Informally, a family of ROM-PRUs is scalable if the security (the adversary query budget and the distinguishing advantage) can be tuned independently of the dimension, and furthermore the PRU is efficient (i.e., polynomial in $\log d$ and $\secp$) to implement.

\subsection{Prior constructions of ROM-PRUs}

We summarize known constructions of ROM-PRUs.

\paragraph{Cryptographically-secure PRUs.} A recent line of works established the existence of \emph{cryptographically-secure} PRUs (i.e., PRUs in the plain model, without a random oracle)~\cite{MPSY24,chen2024efficient,MH25,LQSYZ25,SMLBH25}. Peering into the construction and analysis of these PRUs, however, reveals that they all (essentially) constructed ROM-PRUs first\footnote{In certain cases, such as in the CPFC construction of~\cite{MH25}, the random oracle might have higher-arity outputs (rather than bits) or be a permutation over a large alphabet; we still call these ROM-PRUs for convenience.}, and then instantiated the random oracle with a post-quantum secure PRF.

All of these ROM-PRU constructions have the following parameter settings:
\begin{enumerate}
    \item Construction query complexity and ROM input length $\poly(\log d)$.
    \item No implementation error.
    \item Security against $t$-query adversaries with distinguishing advantage $O(t^c/d)$ for some constant $c \geq 2$.
\end{enumerate}

While the construction is efficient in terms of query complexity and ROM input length, note that the number of queries tolerable by the adversary is at most $d^{1/c} \leq d^{1/2}$. Therefore these constructions are not proved to be scalable; in fact, we show in Section~\ref{sec:entropylb} that these constructions are inherently not scalable.

Is there \emph{any} construction of a scalable ROM-PRU, for any parameter regime? We next present what we call a ``trivial'' construction of a scalable ROM-PRU.

\paragraph{A trivial scalable ROM-PRU.} This trivial construction relies on the existence of, for every dimension $d$ and integer $t$, an \emph{exact} unitary $t$-design $\nu$ (which is both a diamond- and relative-error design) with support size
\[
    |\support(\nu)| = \binom{d^2 + t -1}{d^2 - 1}^2~.
\]
This follows from~\cite[Theorem 11]{roy2009unitary}. In what follows, we will set $t = 2^\secp$. The second item of Proposition~\ref{prop:designdistinguish} implies that $\nu$ is completely indistinguishable from the Haar measure by any quantum algorithm making $t$ queries (even adaptive ones).

We now describe the construction. Here, the random oracle is used as a source of randomness to specify a unitary from $\support(\nu) = \{ U_j \}_j$ (we assume there is a canonical ordering of the unitaries), and the construction algorithm $\cA$ queries the random oracle bit-by-bit to obtain the index $j$, requiring $\log |\support(\nu)|$ queries (one query for each bit of the index). The ROM input length is thus $\log \log |\support(\nu)|$ (the length needed to represent an index into the binary representation of $j$). Once the index $j$ is known, the construction algorithm $\cA$ can directly implement $U_j$, which we assume has been precomputed beforehand (we do not track computational complexity in this setting).

The implementation error is $0$, and the security is against adversaries making $2^\secp$ queries; these cannot distinguish $\nu$ from Haar at all.

This yields a ROM-PRU family (indexed by dimension and the security parameter) with implementation error $0$, security against $2^{\secp}$-query adversaries, with distinguishing advantage $0$. Therefore it would constitute a scalable ROM-PRU family according to the definition, except that the construction query complexity is
\[
    \log |\support(\nu)| = 2 \log \binom{d^2 + t -1}{d^2 - 1} \leq 2 (d^2 - 1)\log \frac{e(d^2 + t - 1)}{d^2 - 1}~.
\]
For $t = 2^\secp$ this is at most $O(d^2 \cdot \secp) = \poly(d,\secp)$, rather than $\poly(\log d,\secp)$.

This leaves open the question of whether there is an efficient construction of a scalable ROM-PRU in terms of the dimension dependence -- if not polynomial in $\log d$, then at least with better than $d^2$ dependence.

\begin{question}
    Is there a scalable ROM-PRU where the construction query complexity $q$ and the ROM input length $m$ satisfy
    \[
        qm \ll d^c \cdot \poly(\log d) \cdot \poly(\secp)
    \]
\end{question}

\subsection{ROM-PRUs, Unitary Designs, and the Unitary Synthesis Problem}
\label{sec:rom-prus-and-designs}

In this section we relate ROM-PRUs with approximate unitary designs, and the unitary synthesis problem of Aaronson and Kuperberg~\cite{aaronson2007quantum}.

First, ROM-PRUs are diamond designs.
\begin{proposition}[ROM-PRUs are diamond designs]
\label{prop:from-rom-pru-to-approx-design}
Let $\mathscr{U} = \{U_f \}_f$ denote a ROM-PRU with security against $t$-query adversaries with distinguishing advantage $\delta$. Then $\mathscr{U}$ is also a $(t,\delta)$-diamond-design.
\end{proposition}
\begin{proof}
	This follows by Definition~\ref{def:rompru} (definition of ROM-PRU) and Proposition~\ref{prop:designdistinguish} (which relates the definition of diamond design to distinguishability by query algorithms).
\end{proof}

Next, we show that an efficient construction of a \emph{scalable} ROM-PRU would imply a positive solution to the unitary synthesis problem. This is a question of whether there is an efficient quantum oracle algorithm that can implement any unitary, given access to a classical oracle that depends on the unitary.
\begin{definition}[Unitary synthesis]\label{def:usyn}
	We say a query algorithm $\cA$ \emph{solves the unitary synthesis problem} for dimension-$d$ unitaries with query complexity $t$, oracle input length $m$, and error $\eps$ if for all unitaries $U \in \unitary(d)$ there exists a boolean function $f: \{0,1\}^{m} \to \{0,1\}$ such that $\cA$ makes $t$ queries to $f$ and implements $U$ up to diamond distance error $\eps$.
\end{definition}

Aaronson and Kuperberg~\cite{aaronson2007quantum} asked the following question:
\begin{question}[Unitary synthesis problem]
Is there an efficient quantum algorithm that solves the unitary synthesis problem with query complexity $t$, oracle input length $m$, and error $1/2$ satisfying $tm \leq \poly(\log d)$?
\end{question}
We note that there are two trivial quantum algorithms for solving unitary synthesis: one with query complexity $t \approx d^2$ and oracle input length $m \approx \log d$ (i.e., reading the classical description of the unitary bit-by-bit), and one with query complexity $t = 1$ and oracle input length $m \approx d^2$ (i.e., using the Bernstein-Vazirani trick as described in~\cite{LMW24}). The best algorithm for unitary synthesis known so far is due to Rosenthal~\cite{R21}, who presented an algorithm with query complexity $t \approx \sqrt{d}$ and oracle input length $m \approx \poly(\log d)$.

Our next observation is that for a sufficiently small $\delta$, a scalable ROM-PRU would imply a solution to the unitary synthesis problem per Definition~\ref{def:usyn}.

\begin{lemma}
    Suppose there exists a scalable ROM-PRU family. Then there exists an algorithm $\cB$ that solves the unitary synthesis problem with error $\frac{1}{2}$ and query complexity and oracle input length $\poly(\log d)$.
\end{lemma}
\begin{proof}
Let $\{ \mathscr{U}_{d,\secp} \}_{d,\secp}$ be the hypothesized scalable ROM-PRU family. Fix a dimension $d$ and set $\secp = c \log d$ for a large constant $c \geq 2$ to be set later. By definition, $\mathscr{U}_{d,\secp} = \{ U_{d,\secp,f} \}_f$ is a ROM-PRU with construction query complexity $q=\poly(\log d)$ and ROM input length $m=\poly(\log d)$, implementation error at most $d^{-c}$, and security against $d^c$-query adversaries with distinguishing advantage at most $d^{-c}$.

Let $\nu$ denote the uniform distribution over $\mathscr{U}_{d,\secp}$. By Proposition~\ref{prop:from-rom-pru-to-approx-design} it follows that $\nu$ is a $(d^c,d^{-c})$-diamond-design. By Lemma~\ref{lem:designs-to-nets} the set $\mathscr{U}_{d,\secp}$ is a $(\eps,\eta)$-net for some $\eta \leq 2/5$ and $\eps = \sqrt{C d^2/d^c}$ for some universal constant $C \geq 1$. Then by \Cref{cor:relaxed-net-to-net} $\mathscr{U}^2_{d,\secp}$ is a $2\eps$-net. Henceforth we omit $d,\secp$ subscripts for clarity.

    Let $V \in \unitary(d)$. Since $\mathscr{U}^2$ is a $2\eps$-net, there exist unitaries $U_{f_1}, U_{f_2} \in \mathscr{U}$ for Boolean functions $f_1,f_2: \{0,1\}^m \to \{0,1\}$ such that $V$ is $2\eps$-close, in diamond distance, to the product $U_{f_2} U_{f_1}$. Consider the following unitary synthesis algorithm $\cB$: it runs the construction algorithm $\cA$ with query access to $f_1$, and then runs $\cA$ again with query access to $f_2$ (we can combine $f_1,f_2$ into a single Boolean function with an extra input bit as a ``switch''). Each invocation of $\cA$ incurs implementation error at most $d^{-c}$, so by the triangle inequality $\cB$ synthesizes a unitary that is $(2\eps+2d^{-c})$-close to $V$. Furthermore, the query complexity is $\poly(\log d)$ and the oracle input length is $\poly(\log d)$.

    We can choose a constant $c \geq 2$ such that for sufficiently large $d$, we have $2\eps + 2d^{-c} \leq \frac{1}{2}$. In fact, by taking $c$ sufficiently large as a function of any fixed $a>0$, the same argument gives synthesis error at most $d^{-a}$, while keeping query complexity and oracle input length $\poly(\log d)$.
\end{proof}

\begin{remark}
    The previous lemma only requires ROM-PRUs that are secure against distinguishers making non-adaptive queries to $U$ (inverse queries $U^\dagger$ and adaptive queries are not needed).
\end{remark}

\begin{remark}
    The reader may note that the scalability property is much more than is necessary; the connection between ROM-PRUs and the unitary synthesis problem holds the moment the ROM-PRU is secure against adversaries making $t \gtrsim d^2$ queries, and the ROM-PRU construction query complexity and ROM input length translate directly to the unitary synthesis solver's query complexity and oracle input length.
\end{remark}

\section{Lower Bounds for Designs and ROM-PRUs}
\label{sec:entropylb}

In this section we present some lower bounds on approximate unitary designs and ROM-PRUs. The lower bounds on ROM-PRUs, in particular, place nontrivial constraints on the space of possible constructions of scalable ROM-PRUs.

We start in Section~\ref{sec:cardlbs} by stating cardinality lower bounds for $t$-designs from previous works, followed by our own contribution for the setting of $t \gg d^2$. Using these bounds, together with the volume bound on $\epsilon$-nets from Corollary~\ref{cor:net-size}, we derive in Section~\ref{sec:functionlb} a lower bound on the size of binary functions required to obtain families of designs and nets. Finally, in Section~\ref{sec:diagonallb}, we show that using diagonal matrices, which supposedly have infinite entropy, does not actually change the picture significantly.

\subsection{Cardinality Lower Bounds for Approximate Designs}
\label{sec:cardlbs}

In this section we prove new lower bounds for the support size of approximate $t$-designs. Prior work~\cite{brandao2016local,brandao2021models} proved the following lower bound: let $\nu$ be a $(t,\delta)$-diamond-design. Then
\begin{equation}
    \label{eq:prior-support-size-bounds}
    \abs{\support(\nu)} \ge \max \Big \{ (1-\delta) \binom{d+t-1}{t}^2  \, , \, \frac{1}{1 + \delta} \, \frac{d^{2t}}{t!} \Big \}~.
\end{equation}
The first part of the bound follows from~\cite[Lemma 26]{brandao2016local}, and the second part follows from~\cite[Lemma 5]{brandao2021models}.

We note that this bound is roughly optimal for $t \leq d$; from~\cite[Theorem 11]{roy2009unitary} we have that there exist \emph{exact} $t$-designs with cardinality at most $\binom{d^2 + t - 1}{t}^2$. Up to a polynomial this is close to the lower bound of $(1 - \delta) \binom{d + t -1}{t}^2$ in the regime that $t \leq d$. However when $t \gtrsim d$, omitting the dependence on $\delta$, this first bound is at most
\[
    \binom{d + t - 1}{t}^2 = \binom{d + t - 1}{d - 1}^2 \leq \Big( \frac{ e(d+t-1)}{d-1} \Big)^{2(d-1)} \leq \Big( \frac{ct}{d} \Big)^{2d}
\]
for some constant $c$. For the regime $d \lesssim t \lesssim d^2$, the first bound is superseded by the second bound, which is at least
\[
    \frac{d^{2t}}{t!} \approx \frac{1}{\sqrt{2 \pi t}} \Big( \frac{ ed^2}{t} \Big)^t
\]
where we used Stirling's approximation of $t!$. However, when $t \gg d^2$, this second bound becomes vacuous. Are there better support size bounds for the regime that $t \gg d^2$? We address this next.

\paragraph{Lower Bounds from Epsilon Net Arguments.} We leverage the connection between approximate $t$-designs and $\eps$-nets that we established in Section~\ref{sec:designstonets}. The approximate design lower bound then follows from well-known lower bounds on the sizes of $\eps$-nets.

\begin{lemma}[Lower Bounds for Approximate Designs for $t \gtrsim d^2$]
\label{lem:improved-support-size-bounds}
	There exists a constant $C>0$ such that the following holds.  Let $\nu$ be a $(t,\delta)$-diamond-design for some $0 \leq \delta < 1$. Then
    \[
    \abs{\support(\nu)} \ge \left( \frac{2 - 2\delta}{3 + \delta} \right) \cdot \left( \frac{Ct}{d^2} \cdot \frac{1}{\ln \frac{1}{(1 - \delta)/4}} \right)^{(d^2-1)/2}~.
    \]
\end{lemma}
\begin{proof}
	Let $\nu$ be a $(t,\delta)$-diamond-design. Set $\eta_0 = (1 - \delta)/4$. Then
    \[
        \eta = \frac{\delta + \eta_0}{1 - \eta_0} = \frac{1 + 3\delta}{3+\delta}
    \]
    By Lemma~\ref{lem:designs-to-nets} there exists a constant $c$ such that $\net = \support(\nu)$ is a $(\epsilon, \eta)$-net with
    \[
        \eps = \sqrt{\frac{c d^2}{t} \ln \frac{1}{\eta_0}}~.
    \]
    The claimed support size bound then follows from Corollary~\ref{cor:net-size}, which gives a lower bound on the size of $(\eps,\eta)$-nets.
\end{proof}

Note that for the regime $t \gtrsim d^2$, with $\delta$ bounded away from $1$, Lemma~\ref{lem:improved-support-size-bounds} gives bounds of size at least $e^{\Omega(d^2)}$, which are better than the bounds~\eqref{eq:prior-support-size-bounds} given in prior work.

\subsection{Entropy Bounds on Constructions with a Classical Oracle}\label{sec:functionlb}

Consider a family of unitaries $\mathscr{U} = \{ U_f \}_f \subseteq \unitary(d)$ where $f: \binset^m \to \binset$ is a classical function.
Since such a family has cardinality at most $2^{2^m}$ we can derive the following conclusions.

\begin{remark}\label{rmk:manyf}
	Without loss of generality, the above also has implications if instead of a single binary function $f$, $U$ has access to functions $g_1, \ldots, g_k$, where $g_i: \binset^m \to \binset^{\ell_i}$. This is equivalent to a single function $f: \binset^{m+\lceil\log \ell\rceil} \to \binset$, where $\ell = \sum_i \ell_i$. This is because $f$ can use $\lceil\log\ell\rceil$ bits to indicate which output bit (of which function $g_i$) it wishes to access.
\end{remark}

We derive the following corollaries by plugging in our cardinality bounds for designs and nets, suppressing lower-order additive constants. The logarithmic expressions below are interpreted in the parameter regimes where they are well-defined.

\begin{corollary}
If $\mathscr{U}$ contains the support of some $(t,\delta)$-diamond-design then
\begin{align}
m \ge \log(t) +  \loglog (d^2 / t) - O(1)~,
\end{align}
and also (assuming $1-\delta=\Omega(1)$):
\begin{align}
	m \ge 2\log(d) + \loglog (t/ d^2) - O(1)~.
\end{align}
\end{corollary}

\begin{corollary}
	If $\mathscr{U}$ is an $(\epsilon,\eta)$-net and $1-\eta = \Omega(1)$, then
	\begin{align}
		m \ge 2\log d + \loglog (1/\epsilon) - O(1)~.
	\end{align}
\end{corollary}

\paragraph{Implications for Existing PRU Constructions and Candidates.} We can derive a conclusion regarding all existing PRU candidates in the oracle model (whether proven or not) \cite{JLS18,MPSY24,MH25,LQSYZ25}. In particular, all such constructions rely on random oracles with input length $\log d$ and output length at most $\polylog d$. They may also use up to $\polylog d$ different functions. It follows, therefore, that such a construction can potentially be a $(t,\delta)$-diamond design only for $t$ at most $d \, \polylog d$. Nor can such constructions be an $\epsilon$-net, or even a relaxed net, for any nontrivial parameters; therefore, they cannot be used as unitary synthesis solutions.

\subsection{Diagonal Matrices Do Not Add Much Entropy}\label{sec:diagonallb}

We show that the use of diagonal matrices with infinite precision does not provide a viable way to get around the entropy barrier. In particular, diagonal matrices with infinite precision can be replaced by classical binary functions of similar dimension.

\begin{theorem}\label{thm:diagonalnogo}
	Let $\{ U^{D_1, \ldots, D_\ell} \}_{D_1, \ldots, D_\ell}$ be a family of unitaries with oracle access to $\ell$ diagonal unitaries of dimension $2^m$, and let $s$ be an upper bound on the number of oracle calls made by $U$ (note that $\ell \le s$ without loss of generality). Then for all $\epsilon$ there exists a classical-oracle implementation family $\{V^f\}_f$, with $f: \binset^{m'} \to \binset$ and $m' \le m + \log s + \loglog ((s/\epsilon)+O(1))$, such that for all $D_1, \ldots, D_\ell$ there exists $f$ satisfying
	\[
	\norm{U^{D_1, \ldots, D_\ell} - V^f}_{\diamond} \le \epsilon~.
	\]
\end{theorem}

We now derive corollaries similar to the above for the diagonal setting. We consider a family of unitaries $\cD = \{ U^{D_1, \ldots, D_\ell} \}_{D_1, \ldots, D_\ell} \subseteq \unitary(d)$ where each $D_i \in \unitary(2^m)$ is a diagonal matrix (note that this family has uncountable cardinality), and let $s = \poly(\log d, \secp)$ for some security parameter $\secp$. The theorem above asserts that this family can be $d^{-\log d}$-approximated by a family of cardinality $2^{2^{m'}}$ for $m' = m + O(\loglog d+\log \secp)$. Hence the following holds.

\begin{corollary}
	If $\cD$ contains the support of some $(t,\delta)$-diamond-design with $1-\delta=\Omega(1)$ then
	\begin{align}
		m \ge \log(t) +  \loglog (d^2 / t) - O(\loglog d+\log \secp)~,
	\end{align}
	and also
	\begin{align}
		m \ge 2\log(d) + \loglog (t/ d^2) - O(\loglog d+\log \secp)~.
	\end{align}
\end{corollary}

\begin{corollary}
	If $\cD$ is an $(\epsilon,\eta)$-net and $1-\eta = \Omega(1)$, then
	\begin{align}
		m \ge 2\log d + \loglog (1/\epsilon) - O(\loglog d+\log \secp)~.
	\end{align}
\end{corollary}

We now prove the theorem.

\begin{proof}[Theorem~\ref{thm:diagonalnogo}]
For a real value $x \in (-1, 1]$ we denote $\rnd{x}_k = 2^{-k} \rnd{2^k x}$, where $\rnd{\cdot}$ is rounding to the nearest integer, breaking ties upwards. We note that $\abs{x - \rnd{x}_k} \le 2^{-(k+1)}$. If $f: X \to (-1, 1]$ is some function, then we define $\rnd{f}_k$ to be the function where $\rnd{f}_k(x) = \rnd{f(x)}_k$. For any function $f: \binset^m \to (-1, 1]$, we let $D_f$ denote the diagonal unitary $D_f = \sum_x e^{i \pi f(x)} \ketbra{x}{x}$. We say that $D_{\rnd{f}_k}$ is the ``$k$-truncation'' of $D_f$. Note that a $k$-truncated function can be expressed as a $\binset^m \to \binset^{k+1}$ function.

It follows from \cite[Definition 1.5, Proposition 1.6]{haah2023query} that
\[
\dnorm{{D_f} - {D_{\rnd{f}_k}}} \le 2^{-k} \pi~,
\] since it holds that $\max_x  \abs{\theta_x} \le \pi 2^{-(k+1)}$.

Let us now consider the family $U^{D_1, \ldots, D_\ell}$. Let $f_i(x)$ be such that $D_i = D_{f_i}$ and let $g_i = \rnd{f_i}_k$, $D'_i = D_{g_i}$. Then by the union bound over the (at most) $s$ oracle calls of $U$:
\[
\norm{U^{D_1, \ldots, D_\ell} - U^{D'_1, \ldots, D'_\ell}}_{\diamond} \le s 2^{-k} \pi ~.
\]
Therefore, by choosing $k = \log(s/\epsilon)+O(1)$ with a sufficiently large additive constant, this distance is at most $\epsilon$.

We point out that each $g_i$ is determined by a function in $\binset^m \to \binset^{k+1}$. By Remark~\ref{rmk:manyf} it therefore holds that $U^{D'_1, \ldots, D'_\ell}$ can be implemented using $V^{f}$, where $f: \binset^{m+\log(s (k+1))} \to \binset$. Plugging in the expression for $k$ from above, the claim follows.
\end{proof}

\section{Optimal distinguisher for PFC}
\label{sec:PFC}

The so-called ``PFC'' ensemble for unitary designs and PRUs was first introduced by~\cite{MPSY24}, which showed that if $P$ is a random $2^n \times 2^n$ permutation matrix, $F$ is a random diagonal unitary matrix, and $C$ is a uniformly random $n$-qubit Clifford operator, then the $t$-th moments of $PFC$ are $O(t/\sqrt{d})$-close in trace distance to those of a Haar-random unitary, where $d = 2^n$. The proof of~\cite{MPSY24} thus only shows that the PFC ensemble is indistinguishable from Haar random when $t \ll \sqrt{d}$. The arguments from Section~\ref{sec:entropylb} imply that PFC cannot be a secure ROM-PRU for $t \gtrsim d$. This left open whether the $t$-th moments of PFC are still indistinguishable from Haar for the range $\sqrt{d} \leq t \leq d$.

In this section we exhibit an optimal distinguisher for the PFC ensemble for $t = \Theta(\sqrt{d})$. That is, $\Theta(\sqrt{d})$ nonadaptive queries to the PFC ensemble are enough to distinguish PFC from Haar random with constant advantage. The distinguisher is as follows. Set parameters
\[
    t = \lceil \sqrt{d} \rceil \, , \qquad \qquad \alpha = \frac{1}{4} \, ,\qquad \qquad k = 100000~.
\]

\begin{enumerate}
    \item Non-adaptively query the unitary oracle $U$ on the all zeroes state, $kt$ times, to obtain $\ket{\psi}^{\otimes kt}$.
    \item Divide the $kt$ copies into $k$ \emph{blocks} of $t$ copies each. For the $r$'th block: Measure all of the $t$ copies of $\ket{\psi}$ in the standard basis to obtain classical strings $x_1^{(r)},\ldots,x_t^{(r)} \in \{0,1\}^n$. Let
    \[
        M^{(r)} = \sum_{i < j} \mathds{1} \{ x_i^{(r)} = x_j^{(r)} \}
    \]
    denote the number of collisions in the $r$'th block.
    \item Let $M = \frac{1}{k} \sum_{r=1}^k M^{(r)}$ denote the empirical average of the number of collisions in each block. Fix a threshold $\alpha$, to be set later. If
    \[
        \Big | M - \binom{t}{2} \frac{2}{d+1} \Big | \leq \alpha
    \]
    then output ``Haar''. Otherwise, output ``PFC''.
\end{enumerate}

The following technical lemma argues that the estimator $M$ concentrates around the average two-way collision probability (meaning measuring two copies of the state $\ket{\psi}$ yields the same string), and the concentration also depends on the three-way collision probability.
\begin{lemma}
\label{lem:pfc-technical}
Let $\ket{\psi}$ be a state with the two- and three-way collision probabilities
    \[
        p_\psi = \sum_{x \in \{0,1\}^n} | \braket{x | \psi } |^4 \qquad \text{and} \qquad q_\psi = \sum_{x \in \{0,1\}^n} | \braket{x | \psi } |^6~.
    \]
    Then defining
    \[
        \mu = \binom{t}{2} p_\psi \qquad \text{and} \qquad \tau = t^2 p_\psi + 2t^3 q_\psi~,
    \]
    we have that for all $\beta > 0$
    \[
        \Pr \Big( \Big | M - \mu \Big| \geq \beta \Big) \leq \frac{\tau}{k \beta^2}
    \]
    where the randomness is over the samples $\{x^{(r)}_j\}_{1 \leq r \leq k, 1\leq j \leq t}$.
\end{lemma}
\begin{proof}
Consider a fixed block $r$. For $1 \leq i < j \leq t$, let $E_{ij}^{(r)}$ denote the indicator variable for whether $x_i^{(r)} = x_j^{(r)}$. Then $M^{(r)} = \sum_{i < j} E_{ij}^{(r)}$ by definition.

In what follows, all probabilities are over the randomness of sampling $\{x_i^{(r)} \}_i$; the state $\ket{\psi}$ is fixed. The expectation $\mu$ of $M^{(r)}$ is
    \[
        \mu = \E M^{(r)} = \sum_{i < j} \E E_{ij}^{(r)}  = \binom{t}{2} p_\psi
    \]

    Now we calculate the variance of $M^{(r)}$ with respect to the randomness of sampling $x_1,\ldots,x_t$.
    \begin{align*}
        \E (M^{(r)})^2 &= \E  \sum_{\substack{i < j \\ i' < j'}} E_{ij}^{(r)} E_{i' j'}^{(r)}~.
    \end{align*}
    We can divide the above sum into three cases: $\{ i, i', j, j' \}$ has two, three, or four distinct indices. Suppose that it is two; then the sum becomes
    \[
        \E \sum_{i < j} E_{ij}^{(r)} = \binom{t}{2} p_\psi~.
    \]
    When there are three distinct indices, the expected value corresponds to the three-way collision probability:
    \[
        \E \sum_{\substack{i < j, i' < j' \\ |\{ i,i',j,j'\}| = 3}} E_{ij}^{(r)} E_{i' j'}^{(r)} =  \sum_{\substack{i < j, i' < j' \\ |\{ i,i',j,j'\}| = 3}} q_\psi \leq  2\binom{t}{2} (t-2) = 6\binom{t}{3} \, q_\psi~.
    \]
    When there are four distinct indices, we have
    \[
        \E \sum_{\substack{i < j, i' < j' \\ |\{ i,i',j,j'\}| = 4}} E_{ij}^{(r)} E_{i'j'}^{(r)} = \sum_{\substack{i < j, i' < j' \\ |\{ i,i',j,j'\}| = 4}} p_\psi^2 = \binom{t}{2}\binom{t-2}{2} p_\psi^2~.
    \]
    Thus the variance of $M^{(r)}$ is
    \begin{align*}
        \sigma^2 &= \binom{t}{2} p_\psi + 6\binom{t}{3} \, q_\psi + \binom{t}{2}\binom{t-2}{2} p_\psi^2 - \binom{t}{2}^2 p_\psi^2 \\
        & \qquad \qquad \leq t^2 p_\psi + 2t^3 q_\psi := \tau
    \end{align*}
    where in the last line we used that $\binom{t-2}{2} \leq \binom{t}{2}$.

    The variance of the average number of collisions in each block is thus
    \[
        \mathrm{Var} \Big(\frac{1}{k} \sum_{r=1}^k M^{(r)} \Big) \leq \frac{\sigma^2}{k} \leq \frac{\tau}{k}~.
    \]
    Thus, as we average the number of collisions across the $k$ blocks, we obtain an estimate that is concentrated around the mean $\mu$. By Chebyshev's inequality the probability that the average number of collisions $M = \frac{1}{k}\sum_r M^{(r)}$ deviates from $\mu$ by more than $\beta$ is bounded by
    \[
        \Pr \Big( \Big | M - \mu \Big| \geq \beta \Big) \leq \frac{\tau}{k \beta^2}~.
    \]
\end{proof}

Now we argue that, for an appropriate choice of $\alpha$, the estimator $M$ for a Haar-random state $\ket{\psi}$ is close to $\binom{t}{2} \frac{2}{d+1}$ with high probability.
\begin{lemma}
\label{lem:pfc-haar-case}
For sufficiently large $d$, if $U$ is sampled from the Haar measure, then the distinguisher outputs ``Haar'' with probability at least $0.99$.
\end{lemma}
\begin{proof}
    When $U$ is Haar-random, the state $\ket{\psi} = U\ket{0 \cdots 0}$ is Haar-random. In what follows, the probabilities will be taken over the randomness of sampling $\ket{\psi}$.

    \paragraph{Expectation.} Note that
    \begin{align*}
        \E_\psi p_\psi &= \E_\psi \sum_x \braket{x,x | \psi, \psi} \braket{\psi,\psi | x,x} = \sum_x \bra{x,x} \frac{\Pi_{\mathrm{sym}}^{(2)}}{\Tr(\Pi_{\mathrm{sym}}^{(2)})} \ket{x,x} \\
        &= \frac{d}{\binom{d+1}{2}} = \frac{2}{d+1}~.
    \end{align*}
    We used that averaging over two copies of a Haar-random state yields the maximally mixed state on the symmetric subspace (whose projector is $\Pi_{\mathrm{sym}}^{(2)}$).

    Similarly, we have
    \begin{align*}
        \E_\psi q_\psi &= \sum_x \bra{x,x,x} \frac{\Pi_{\mathrm{sym}}^{(3)}}{\Tr(\Pi_{\mathrm{sym}}^{(3)})} \ket{x,x,x} \\
        &= \frac{d}{\binom{d+2}{3}} = \frac{6}{(d+2)(d+1)}~.
    \end{align*}

    \paragraph{Variance.} Note that
    \begin{align*}
        \E_\psi p_\psi^2 &= \E_\psi \sum_{x,y} | \braket{x | \psi } |^4 \cdot | \braket{y | \psi } |^4 = \E_\psi \sum_{x,y} \braket{x,x,y,y | \psi,\psi,\psi,\psi}\braket{\psi,\psi,\psi,\psi | x,x,y,y} \\
        &= \frac{1}{\binom{d+3}{4}} \sum_{x,y} \bra{x,x,y,y} \Pi_{\mathrm{sym}}^{(4)}\ket{x,x,y,y}
    \end{align*}
    where $\Pi_{\mathrm{sym}}^{(4)}$ denotes the projector onto the four-fold symmetric subspace. We divide the sum into two cases. When $x = y$, the sum becomes
    \[
        \frac{d}{\binom{d+3}{4}} = \frac{4!}{(d+3)(d+2)(d+1)}
    \]
    When $x \neq y$, we have
    \[
        \Pi_{\mathrm{sym}}^{(4)}\ket{x,x,y,y} = \frac{1}{\sqrt{6}} \Big( \ket{x,x,y,y} + \ket{y,y,x,x} + \ket{x,y,x,y} + \ket{y,x,y,x} + \ket{x,y,y,x} + \ket{y,x,x,y} \Big)~.
    \]
    Therefore $\bra{x,x,y,y} \Pi_{\mathrm{sym}}^{(4)}\ket{x,x,y,y} = \frac{1}{6}$. Putting everything together, we have
    \[
        \E_\psi p^2_\psi = \frac{4! + 4(d-1)}{(d+3)(d+2)(d+1)}
    \]
    Thus the variance is
    \begin{align*}
        \mathrm{Var}_\psi( p_\psi) = \E_\psi p^2_\psi - \Big( \E_\psi p_\psi \Big)^2 &= \frac{4! + 4(d-1)}{(d+3)(d+2)(d+1)} - \frac{4}{(d+1)^2} \\ &= \frac{4(d+5)(d+1) - 4(d+3)(d+2) }{(d+3)(d+2)(d+1)^2} \\
        &= \frac{4(d-1)}{(d+3)(d+2)(d+1)^2} = \Theta(d^{-3})~.
    \end{align*}

    Define $\ket{\psi}$ to be \emph{good} if
    \begin{enumerate}
        \item $\Big| p_\psi - \frac{2}{d+1} \Big| \leq \frac{1}{100d}$
        \item $q_\psi \leq d^{-3/2}$.
    \end{enumerate}
    A state $\ket{\psi}$ fails to be good only if either condition (1) or condition (2) fails to hold. By Chebyshev's inequality, we have
    \[
        \Pr_\psi \Big( \Big| p_\psi - \frac{2}{d+1} \Big| \geq \frac{1}{100d} \Big) \leq \mathrm{Var}_\psi( p_\psi)\cdot O(d^2) = O(d^{-1})~.
    \]
    Thus condition (1) fails to hold with probability at most $O(1/d)$. By Markov's inequality condition (2) fails to hold with probability at most $O(d^{-1/2})$. Therefore $\ket{\psi}$ is not good with probability at most $O(d^{-1/2})$.

    Fix a good $\ket{\psi}$. Now consider the distinguisher.
    The distinguisher is checking whether the empirical average $M$ deviates from $\binom{t}{2} \frac{2}{d+1}$ by more than $\alpha$; this implies the event that $M$ deviates from $\mu := \binom{t}{2} p_\psi$ by more than $\alpha - |\mu - \binom{t}{2} \frac{2}{d+1}| \geq \alpha - \binom{t}{2} \frac{1}{100d} \geq \frac{1}{8} =: \beta$ (this uses that $\ket{\psi}$ is good and that $d$ is sufficiently large). We can thus bound the probability as follows:
    \[
       \Pr \Big( \Big | M - \binom{t}{2} \frac{2}{d+1} \Big| \geq \alpha \, : \, \psi \text{ good}\Big) \leq \Pr \Big( \Big | M - \mu \Big| \geq \beta \, : \, \psi \text{ good} \Big) \leq \frac{\tau}{k \beta^2}
    \]
    where we used Lemma~\ref{lem:pfc-technical}.

    Then we have that
    \begin{align*}
        \Pr \Big( \text{distinguisher outputs ``Haar''} \Big) &\geq \Pr( \psi \text{ is good}) \Big(1 - \frac{\tau}{k \beta^2} \Big) \\
        &\geq (1 - O(d^{-1/2})) \Big(1 - \frac{\tau}{k \beta^2} \Big)~.
    \end{align*}

    Now we compute an upper bound on $\tau := t^2 p_\psi + 2t^3 q_\psi$. Since $\ket{\psi}$ is good, we have that $p_\psi \leq \frac{2}{d+1} + \frac{1}{100d} \leq \frac{3}{d}$ and $q_\psi \leq d^{-3/2}$. Thus
    \[
        \tau \leq \frac{3t^2}{d} + \frac{2t^3}{d^{3/2}} \leq 5
    \]
    for all sufficiently large $d$. Thus for sufficiently large $d$, we have
    \[
        \Pr \Big( \text{distinguisher outputs ``Haar''} \Big) \geq 1 - \frac{320}{k} - O(d^{-1/2}) \geq 0.99
    \]
    as claimed.
\end{proof}

On the other hand, we show that the PFC ensemble is detected with probability at least $\frac{1}{4}$ by the distinguisher.
\begin{lemma}
\label{lem:pfc-pfc-case}
    For sufficiently large $d$, if $U$ is sampled from the PFC ensemble, then the distinguisher outputs ``PFC'' with probability at least $1/4$.
\end{lemma}
\begin{proof}
    Let $U = PFC$ for a random permutation matrix $P$, random diagonal unitary $F$, and random Clifford $C$. Let $\ket{\varphi} = C\ket{0 \cdots 0}$. By definition $\ket{\varphi}$ is a stabilizer state; for every fixed choice of $C$ there exists an affine subspace $A \subseteq \mathbb{Z}_2^n$ such that measuring $\ket{\varphi}$ in the standard basis yields a uniformly random element of $A$:
    \[
        \ket{\varphi} = \frac{1}{\sqrt{|A|}} \sum_{x \in A} \alpha_x \ket{x}
    \]
    where $\alpha_x$ are complex numbers on the unit circle. The state $\ket{\psi} = PFC\ket{0 \cdots 0} = PF \ket{\varphi}$ has the property that measuring it in the standard basis yields the uniform distribution over a set of size $|A|$ (though the set need not be a subspace, because of the permutation $P$).

    For a fixed $\ket{\psi}$, the two- and three-way collision probabilities are straightforward to calculate:
    \[
        p_\psi = \frac{1}{|A|} \qquad \text{and} \qquad q_\psi = \frac{1}{|A|^2}~.
    \]
    Condition on the event that $|A| = d$; in other words, the stabilizer state has full support. Then $\mu := \binom{t}{2} \frac{1}{d} \leq \frac{t^2}{2d}$ and
    \[
        \tau = \frac{t^2}{d} + \frac{2t^3}{d^2} \leq 2~.
    \]
    By Lemma~\ref{lem:pfc-technical} we have that for all $\beta > 0$,
    \[
        \Pr \Big( \Big| M - \mu \Big| \geq \beta : |A| = d \Big) \leq \frac{\tau}{k \beta^2}~.
    \]
    Set $\beta := \frac{1}{8}$. Since $\binom{t}{2} \frac{2}{d+1} - \mu = \binom{t}{2} \Big( \frac{2}{d+1} - \frac{1}{d} \Big) \geq \frac{1}{3}$ for sufficiently large $d$, we have that
    \[
        \Pr \Big( \Big| M - \binom{t}{2} \frac{2}{d+1} \Big| \leq \alpha : |A| = d \Big) \leq \Pr \Big( \Big| M - \mu \Big| \geq \beta : |A| = d \Big) \leq \frac{\tau}{k \beta^2}~.
    \]
    This means that
    \begin{align*}
        &\Pr(\text{distinguisher outputs ``PFC''}) \\
        &\qquad \qquad \geq \Pr_\psi(|A| = d) \Pr \Big( \Big| M - \binom{t}{2} \frac{2}{d+1} \Big| \geq \alpha : |A| = d \Big) \\
        &\qquad \qquad \geq \Pr_\psi(|A| = d) \Big( 1 - \frac{\tau}{k\beta^2} \Big) \geq \Pr_\psi(|A| = d) \Big (1 - \frac{128}{k} \Big)~.
    \end{align*}

    Finally, we need to lower bound the probability that $|A| = d$. We do this by exhibiting a large family of stabilizer states with full support. Let $M \in \{0,1\}^{n \times n}$ be a matrix with zeros along the diagonal, and let $u,v \in \{0,1\}^n$. Define the state
    \[
        \ket{\Gamma_{M,u,v}} = 2^{-n/2} \sum_{x \in \{0,1\}^n} \iu^{u \cdot x} (-1)^{x^\top M x + v \cdot x} \ket{x}~.
    \]
    This is a stabilizer state because it is obtained from $\ket{+}^{\otimes n}$ by the following Clifford circuit:
    \begin{enumerate}
        \item Apply $S$ to qubit $i$ if $u_i = 1$.
        \item Apply $CZ_{ij}$ if $M_{ij} = 1$.
        \item Apply $Z$ to qubit $i$ if $v_i = 1$.
    \end{enumerate}
    These states are distinct for every choice of $M,u,v$.
    Counting, there are $2^{n(n-1)/2} \cdot 2^{2n}$ such states. On the other hand, it is well-known that the total number of $n$-qubit stabilizer states is
    \[
        2^n \prod_{j =1}^n (2^j + 1)~.
    \]
    Taking ratios we get that the number of full-support stabilizer states is at least
    \[
        \Pr_\psi (|A| = d) \geq \frac{2^{n(n-1)/2} \cdot 2^{2n}}{2^n \prod_{j =1}^n (2^j + 1)} = \prod_{j=1}^n \frac{1}{1+2^{-j}} \geq \exp \Big( -\sum_{j=1}^n 2^{-j} \Big) \geq e^{-1}~.
    \]

    Putting everything together, we have that
    \[
        \Pr(\text{distinguisher outputs ``PFC''}) \geq \frac{1}{e} \Big( 1 - \frac{128}{k} \Big) \geq \frac{1}{4}
    \]
    as claimed.
\end{proof}

\begin{lemma}
    There exists an adversary for the PFC ensemble that makes $\Theta(\sqrt{d})$ queries and distinguishes PFC from Haar with constant advantage.
\end{lemma}
\begin{proof}
    This follows from Lemma~\ref{lem:pfc-haar-case} and Lemma~\ref{lem:pfc-pfc-case}.
\end{proof}

\section*{Acknowledgments}

Zvika Brakerski is supported by the Horizon Europe Research and Innovation Program via ERC Project ACQUA (Grant 101087742). Henry Yuen is supported by AFOSR
award FA9550-23-1-0363, NSF awards CCF-2530159, CCF-2144219, and CCF-2329939, and by the
Sloan Foundation.

\printbibliography

\end{document}